\documentclass[runningheads]{llncs}

\usepackage{amsmath, amsfonts, amssymb, amsbsy}
\usepackage{color, colortbl, tcolorbox}
\usepackage{url, array, cite, multirow, mathtools, framed, newfloat,microtype,algpseudocode,algorithm}
\usepackage[OT1]{fontenc}
\usepackage{diagbox, booktabs, appendix, alltt}
\usepackage{listings, xcolor, afterpage, bm, xspace}
\usepackage[shortlabels]{enumitem}
\usepackage[multiple]{footmisc}
\usepackage[cal=boondox]{mathalfa}
\usepackage[capitalise]{cleveref}
\usepackage{caption}
\usepackage{derivative}
\usepackage[italicdiff]{physics}
\usepackage[utf8]{inputenc}

\newcommand{\tx}{\mathsf{tx}\xspace}
\newcommand{\Tx}{\mathsf{Tx}\xspace}
\newcommand{\OPT}{\mathrm{OPT}}
\newcommand{\A}{\mathrm{A}}
\newcommand{\FA}{\textsc{FlushAll}}

\newcommand{\FWF}{\textsc{FlushWhenFull}}
\newcommand{\FTWF}{\textsc{FlushTwoWhenFull}}
\newcommand{\FWFabbr}{\textsc{FWF}}
\newcommand{\D}[3]{\CostF{#3}{\MDisc{#1}{#2}}{\Tx}}
\newcommand{\G}[3]{\CostF{#3}{\MGen{#1}{#2}}{\Tx}}

\newcommand{\CostF}[3]{#1[#2; #3]}
\newcommand{\Cost}[2]{#1[#2]}
\newcommand{\MDisc}[2]{\mathsf{Disc}^{#1}_{#2}}
\newcommand{\MGen}[2]{\mathsf{Gen}^{#1}_{#2}}

\begin{document}

\title{Competitive Policies for Online Collateral Maintenance}

\author{Ghada Almashaqbeh\inst{1} \and Sixia Chen\inst{2} \and Alexander Russell\inst{1,3}}

%
%
\institute{University of Connecticut, \email{\{ghada, acr\}@uconn.edu}\\
\and
Adelphi University, \email{schen@adelphi.edu}
\and
IOG}

\maketitle

\begin{abstract}
Layer-two blockchain protocols emerged to address scalability issues related to fees, storage cost, and confirmation delay of on-chain transactions. They aggregate off-chain transactions into fewer on-chain ones, thus offering immediate settlement and reduced transaction fees. To preserve security of the underlying ledger, layer-two protocols often work in a collateralized model; resources are committed on-chain to backup off-chain activities. A fundamental challenge that arises in this setup is determining a policy for establishing, committing, and replenishing the collateral in a way that maximizes the value of settled transactions.

In this paper, we study this problem under two settings that model collateralized layer-two protocols. The first is a general model in which a party has an on-chain collateral $C$ with a policy to decide on whether to settle or discard each incoming transaction. The policy also specifies when to replenish $C$ based on the remaining collateral value. The second model considers a discrete setup in which $C$ is divided among $k$ wallets, each of which is of size $C/k$, such that when a wallet is full, and so cannot settle any incoming transactions, it will be replenished. We devise several online policies for these models, and show how competitive they are compared to optimal (offline) policies that have full knowledge of the incoming transaction stream. To the best of our knowledge, we are the first to study and formulate online competitive policies for collateral and wallet management in the blockchain setting.
\end{abstract}

\section{Introduction}
\label{sec:intro}
Distributed ledger technology has provided a financial and
computational platform realizing an unprecedented combination of trust
assumptions, transparency, and flexibility. Operationally, these
platforms introduce two natural sources of ``friction'': settlement
delays and settlement costs. The Bitcoin protocol, for example,
provides rather lackluster performance in both dimensions, with
nominal settlement delays of approximately one hour and average fees
of approximately 1 USD per transaction. Layer-two protocols have been
the ready response to these complaints as they can provide instant
settlement and, furthermore, can significantly reduce transaction
costs by aggregating related off-chain transactions so that they
ultimately correspond to fewer underlying ledger, or on-chain,
transactions. Examples of such protocols include payment channels and
networks~\cite{Decker15,Poon15}, probabilistic
micropayments~\cite{Pass15,Chiesa17,microcash}, state channels and
networks~\cite{chakravarty2020hydra,dziembowski2018general,miller2019sprites},
and rollups~\cite{Poon17,gluchowski2019zk}.

However, in order for layer-two protocols to provide these remarkable
advantages without sacrificing the security guarantees of the
underlying ledger, they must \emph{collateralize} their activities. In
particular, there must be resources committed on-chain that provide
explicit recourse to layer-two clients in the event of a malicious or
faulty layer-two peer or server. Moreover, the total value of the
on-chain collateral must scale with the value of ``in flight''
transactions supported by the layer-two protocol.

These considerations point to a fundamental challenge faced by
layer-two protocols: determining a policy for establishing,
committing, and replenishing the collateral. Such a policy must ensure
sufficient available collateral to settle anticipated transaction
patterns while minimizing the total collateral and controlling the
resulting number of on-chain transactions. Of course, any fixed
collateralization policy can be frustrated by the appearance of an
individual transaction---or a sudden burst of transactions---that
exceeds the total current collateral. More generally, it would appear
that designing a satisfactory policy must rely on detailed
information about future transaction size and frequency, i.e., transaction distribution. From a
practical perspective, this poses a serious obstacle because
real-world transaction patterns are noteworthy for their
unpredictability and mercurial failure to adhere to a steady
state. Analytically, this immediately calls in to question the value
of distribution-specific solutions. These considerations motivate us
to elevate \emph{distribution independence} as a principal design
consideration for collateral policies.

We formulate a distribution-independent approach by adapting to our
setting the classical framework of competitive analysis. In
particular, we study two natural models: the \emph{$k$-wallet} model
in which the total collateral $C$ is divided among $k$ wallets of
fixed size, and a general model in which $C$ is viewed as one wallet
that allows replenishment of any portion of $C$. After fixing only two
parameters of the underlying system---the total collateral $C$ and the
size $T$ of the largest transaction that we wish to support---we
measure the performance of a given collateral policy against the
performance of an \emph{optimal, omniscient} policy. This optimal
policy utilizes the same total collateral, but has full knowledge of
the future sequence of transactions as it commits and replenishes
collateral. Naively, this would appear to be an overly ambitious
benchmark against which to measure an algorithm that must make choices
on the fly based only on the past sequence of transactions. Our
principal contribution is to show that the natural policies for these
two models perform well, even when compared against this high bar.

\subsection{Contributions}
Our formal modeling is intended to reflect the challenges faced by
standard layer-two protocols. The most immediate of the models we
consider arises as follows: Consider a layer-two protocol with a total
of $C$ collateral that must serve an unknown transaction sequence
$\Tx = (\tx_1, \tx_2, \ldots)$. As each transaction arrives, the
policy may either commit a corresponding portion of its available
collateral to \emph{settle} this transaction or simply \emph{discard}
it; in particular, in any circumstances where there isn't sufficient
uncommitted collateral to cover a given transaction, the transaction
must be discarded. The policy may also---whenever it
chooses---replenish its currently committed collateral. This ``flush''
procedure returns the committed collateral to the available pool of
collateral after a fixed time delay $F$ and involves a fixed cost
$\tau$ (so transactions arriving during $F$ will be discarded if no
other sufficient collateral is available). Thus, the challenge is to
schedule the flush events so as to minimize the total cost while
simultaneously maximizing the total value of settled transactions.

We remark that transactions ``discarded'' in the model above would
typically be handled by some other fallback measure in a practical
setting. The flush operation, in practice, corresponds to on-chain
settlement of a family of transactions that releases the associated
collateral so that it can be reused as surety for additional
transactions. While we assume that the flush procedure is associated
with a fixed, constant cost for simplicity, in practice this cost may
scale with the complexity of the aggregated transactions. We remark
that a fixed cost directly models Lightning-like payment channels and
networks, or escrow-based probabilistic micropayments, where the total
number of participants is bounded.\footnote{On-chain transaction cost
  also varies based on network conditions; during periods of high
  activity or congestion, transaction issuers may resort to increasing
  transaction fees to incentivize miners to prioritize their
  transactions. As such, $\tau$ above is viewed as the average
  transaction cost.}

In this general setting, we study the natural family of policies
determined by a parameter $\eta \in (0,1)$ that settle transactions as
they arrive until an $\eta$-fraction of all collateral is consumed; at
this point the committed collateral is flushed and the process is
continued with the remaining collateral. Our analytic development
first focuses on a simpler variation---of interest in its own
right---that we call the \emph{$k$-wallet} problem. As above, the
policy is challenged to serve a sequence $\Tx$ of transactions with a
total of $C$ collateral; however, the collateral is now organized into
$k$ wallets, each holding $C/k$ collateral, with the understanding
that an entire wallet must be flushed at once. When a wallet is
flushed it becomes entirely unavailable for settlement---regardless of
how much of the wallet was actually committed to settled
transactions---until the end of the flush period $F$, when the
collateral in the wallet is again fully available for future
settlement. As above, the policy may settle a transaction by
committing a portion of collateral in one of the wallets corresponding
to the size of the transaction. This version of the problem has the
advantage that performance is captured by a single quantity: the total
value of settled transactions.

\subsection{A Survey of the Results} 
Continuing to discuss the $k$-wallet model, we consider a sequence
$\Tx$ of transactions, each of value no more than $T$. We focus on the
natural $\FWF$ policy, which maintains a \emph{single active wallet}
(unless all wallets are currently unavailable) that is used to settle
all arriving transactions; if settling a transaction would leave
negative residual committed collateral in the active wallet, the
wallet is flushed and a new wallet is activated as soon as one becomes
available. We prove that this simple, attractive policy settles at
least a fraction
\[
  \frac{1 - kT/C}{1 + 1/k}
\]
of the total value settled by an optimal, offline strategy with $C$
collateral, even one that is not restricted to a $k$-wallet policy but
can flush any portion of its collateral at will. We remark that this
tends to optimality for large $k$ and small $T < C/k$. This result
also answers a related
question: that of how many wallets one should choose for a given total
collateral $C$ and maximum transaction size $T$. We find that optimal
$k$ in this case is $\approx \sqrt{1 + C/T}-1$.

As for the more flexible setting---under the general $C$ collateral model---where the policy may flush any
portion of its collateral at will by paying a transaction fee $\tau$,
recall that this poses a bicriteria challenge: maximizing settled
transactions while reducing settlement fees. We study this by
establishing the natural figure of merit that arises by assuming that
each settled transaction yields positive utility to the policy that
scales with its value (e.g., a ``profit margin''). Thus, the policy
seeks to maximize $p V - \tau f$, where $V$ is the total value of
settled transactions, $f$ is the total number of flushes, and $p$ is
the profit margin. Here we study the family of policies that flush
when currently committed (but unflushed) collateral climbs to an
$\eta$-fraction of $C$ ($\eta$ is a policy parameter). We find
that this policy achieves total utility of at least $1/\alpha$
fraction of that achieved by the optimal omniscient policy, where
\[
  \alpha =  \frac{1}{1-\eta  - T/C} \cdot \frac{p/\tau-1/C}{p/\tau-1/(\eta C)}\,.
\]
In this case, we are also able to determine the optimal constant $\eta^*$ (as a function of $C$, $p$, and $\tau$) that maximizes the policy utility:
\[
  \eta^* = \sqrt{(1 - T/C)\cdot \tau/(pC)}\,.
\]
We remark that our results in the $k$-wallet setting can also be
applied to directly yield results with this accounting that assigns a
flush cost and a profit margin.

We study some additional questions that arise naturally. For example,
we show that no deterministic, single wallet policy can be competitive
if the maximum transaction size can be as large as the wallet size and
show that, on the other hand, a natural randomized algorithm is
$O(1)$-competitive.

\medskip
 
\subsection{Applications} 
Online collateral management arises in various layer-two protocols, as
well as in Web 3.0 and decentralized finance (DeFi) applications. For
layer-two protocols, payment networks are an emblematic example: A
relay party creates payment channels with several parties, allowing
her to relay payments over multi-hop routes. Each payment channel is
tied to a collateral $C$ such that the relay cannot accept a
transaction to be relayed if the remaining collateral cannot cover
it. This applies as well to state channels, where transactions created
off-chain---while the channel is active---are accepted only if their
accumulated value does not exceed the initial fund committed when the
channel was created. These configurations adhere to the general
collateral model discussed above.

Probabilistic micropayments follow a slightly different
setting. Micropayments are usually used to permit service-payment
exchange without a trusted party to reduce financial risks in case of
misbehaving entities. A client creates an escrow fund containing the
collateral backing all anticipated payments to a set of servers. A
server provides a service to the client (e.g., file storage or content
distribution) in small chunks, so that the client pays a micropayment
for each chunk. For any incoming service exchange, the client cannot
take it unless her collateral can pay for it. The client can decide to
replenish the escrow fund to avoid service interruption, thus this
also follows the general collateral model. The client may also choose
to divide her collateral among several escrows, each of which has a
different or similar setting with respect to, e.g., the set of servers
who can be paid using an escrow and the total service payment
amount. This configuration follows the $k$-wallet model.

Apart from layer-two scalability solutions, online collateral
management captures scenarios related to Web 3.0 and DeFi
applications. The framework of decentralized resource markets build
systems that provide digital services, e.g., file storage, content
distribution, computation outsourcing or video transcoding, in a fully
decentralized way~\cite{filecoin,livepeer,golem}. Due to their
open-access nature, where anyone can join the system and serve others,
these systems usually involve some form of collateral. In this case, a
collateral represents the amount of service a party wants to pledge in
the system. For example, in Filecoin~\cite{filecoin}---a distributed
file storage network---a storage server commits collateral proportional
to the amount of storage she claims to own. This server cannot accept
more file storage contracts, and subsequently more storage payments,
than what can be covered by the pledged storage (or alternatively
collateral).

In the DeFi setting, online collateral management is encountered in a
variety of applications. Loan management is a
potential
example~\cite{gan2022understanding,saengchote2023decentralized};
incoming loan requests cannot be accepted unless the loan funding pool
can support them. The loan DeFi application then has to decide a
policy for loan request accept criteria (to favor some requests over
others under the limited funding constraint) and when to replenish the
loan pool balance.

Another potential application, of perhaps an extended version of our
models and policies, that we believe to be of interest is the case of
automated market makers (AMMs)~\cite{xu2023sok}. Here, a liquidity
pool trades a pair of tokens against each other, say token $A$ and
token $B$, such that a trade buying an amount of token $A$ pays for
that using an amount of token $B$, and vice versa. Incoming trades are
accepted only if the liquidity pool can satisfy them, so in a sense
having tokens that can serve the requested trades is the
collateral. Replenishing the pool fund, or liquidity, can be done
organically based on the trades. That is, a particular trade, say to
buy $A$ tokens, reduces the backing fund of token $A$ while increasing
it for token $B$. Another approach for pool replenishment is via
liquidity providers; particular parties provide their tokens to the
pool to serve incoming trades (or token swaps) in return for some
commission fees. These providers can configure when their offered
liquidity can be used, i.e., at what trading price range, under what
is called concentrated liquidity as in Uniswap~\cite{uniswap}. An
interesting open question is to develop competitive collateral
policies that capture this setting where settling a transaction does
not only depend on whether the remaining collateral $C$ (i.e., pool
liquidity) can cover it, but also on transaction-specific parameters
to meet certain collateral-related conditions. Even the replenishment
itself, i.e., providing liquidity, could be subject by other factors
such as the resulting price slippage, so an incoming mint transaction
(in the language of AMMs) that provides liquidity may not be accepted
immediately. We leave these questions as part of our future work
directions.

In general, our work lays down foundations for wallet management to
address issues related to robustness, availability, and profitability
of the wallet(s) holding the collateral. Maintaining one wallet may
lead to periods of interruption; a party must wait for a while before
a new wallet is created to replace an older expired one. Maintaining
several wallets may help, but given the cost of locking currency in a
wallet or renewing it, the number of active wallets and their
individual balances must be carefully selected. Moreover, under this
multi-wallet setting, it is important to consider how incoming
transactions are matched to the wallets, and whether factors such as
payment amount or frequency may impact this decision. A potential
extension to our model is considering adaptive policy management,
where the size of the collateral and the number of wallets can be
adjusted after each flush decision to account for these varying
factors.


\section{The Model; Measuring Policy Quality}
\label{sec:model}

As discussed above, we consider the problem of designing an online
collateral management policy in which a collateral fund of initial
value $C$ is used to settle transactions---each with a positive real
value no more than $T$---chosen from a sequence
$\textsf{Tx} = (\textsf{tx}_1, \textsf{tx}_2, \ldots)$. Operationally,
the policy is presented with the transactions one-by-one and, as each
transaction arrives, it must immediately choose whether to
\emph{settle} the transaction or \emph{discard} it. Settling a
transaction requires committing a portion of the collateral equal to
the value of the transaction; such committed collateral cannot be used
to settle future transactions. Of course, if there isn't sufficient
uncommitted collateral remaining to settle a given transaction when it
arrives, the transaction must be discarded. Committed collateral may
be returned to service by an operation we call a \emph{flush}; we
focus on two different conventions for the flush operation, described
below, but in either case the collateral only becomes available for
use after a fixed time delay $F$.

We assess the performance of a particular online policy $\A$ against
that of an \emph{optimal offline policy} $\OPT$ that knows the full
sequence $\textsf{Tx}$ and can make decisions based on this knowledge.

Below, we describe two models for the collateral: the discrete $k$-wallet model and the general collateral model.

\subsection{The Discrete $k$-Wallet Model}
The $k$-wallet model calls for the collateral to be divided into $k$
wallets, each with $C/k$ collateral value. Wallets support two
operations: (i) a wallet with uncommitted collateral $R$ may
immediately settle any transaction $\tx$ of value $v \leq R$; this
reduces the available collateral of the wallet to $R - v$, and (ii) a
wallet may be flushed, which takes the wallet entirely offline for a
flush period $F$ after which the available collateral $R$ is reset to
$C/k$. As a matter of bookkeeping, we mentally organize time into
short discrete slots indexed with natural numbers: we then treat the
transaction $\tx_t$ as arriving at time(slot) $t$, and set $\tx_t = 0$
for times $t$ when no transactions arrive. We treat the flush
period as a half-open and half-closed interval: if a wallet flushes at
time $t$, then it is offline during the time interval $(t, t+F]$.  In
this model, the figure of merit is the \emph{total value} of settled
transactions. We let $\MDisc{C,k}{T}$ denote this discrete $k$-wallet
model with maximum transaction size $T$.\medskip

\noindent\textbf{Settlement algorithms, settled value, and the competitive
  ratio.} A \emph{$k$-wallet settlement algorithm $\A$} is an algorithm
that determines, for any transaction sequence $\Tx$, whether to
settle each transaction, which wallet to use, and when to flush each
wallet. For such an algorithm $\A$ and a sequence
$\Tx = \tx_1, \tx_2, \ldots, \tx_n$ we let
$\CostF{\A}{\MDisc{C,k}{T}}{\Tx}$ denote the total value of all
transactions settled by the algorithm. In general, we use the notation
$\CostF{\A}{\mathcal{M}}{\Tx}$ to denote the value achieved by
algorithm $\A$ in model $\mathcal{M}$ with input sequence $\Tx$. When
the model is clear from context, we simply write $\Cost{\A}{\Tx}$.

We say that an algorithm $\A$ is \emph{online} if, for every $N$, any
decisions made by the algorithm at time $N$ depend only on 
$\tx_1, \tx_2, \ldots, \tx_N$, i.e., transactions seen so far. We let $\OPT$ denote the optimal
(offline) policy; thus $\CostF{\OPT}{\MDisc{C,k}{T}}{\Tx}$ denotes the
maximum possible value that can be achieved by any policy, even one
with a full view of all (past and future) transactions.

\begin{definition}
  We say that an algorithm $\A$ is \emph{$\alpha$-competitive} in the
  $k$-wallet model if, for any sequence $\Tx = \tx_1, \ldots, \tx_n$
  with maximum value no more than $T$,
  \[
    \CostF{\OPT}{\MDisc{C,k}{T}}{\Tx} \leq \alpha \cdot \CostF{\A}{\MDisc{C,k}{T}}{\Tx}  + O(1)\,,
  \]
  where the constant in the asymptotic notation may depend on the
  model parameters ($C$, $k$, and $T$), but not the sequence
  $\Tx$ or its length $n$.
\end{definition}

\begin{remark}[Relation to the bin packing problem]
  We remark on the relationship between our problem and the
  well-studied online bin packing
  problem~\cite{seiden2002online,darapuneni2012survey}, where an
  algorithm must pack arriving objects into bins of constant size,
  while \emph{opening} a new bin any time a newly arriving object does
  not fit into any of the current bins. In this context, the
  $k$-wallet model calls for a bounded number of \emph{bins} (a.k.a.,
  wallets) that can only be reset with the flush operation. Also, we
  measure the total settled value rather than the number of utilized
  bins. In any case, we adopt the standard classical paradigm of
  competitive analysis to study our algorithms, as described previously.
\end{remark}

\subsection{The General Collateral Model} In contrast to the discrete
$k$-wallet model, where each wallet must be flushed as a whole, the
general setting permits any portion of the collateral to be flushed at
any time. The basic framework is identical: the policy is presented with a
sequence of transactions $\tx_1, \tx_2, \cdots$ and must decide
whether each transaction will be settled or discarded; the total
collateral $C$ and the maximum transaction size $T$ are parameters of
the problem. Settling a transaction requires committing collateral of
value equal to the transaction; however, any portion of the committed
collateral can be flushed at any time. As before, each flush period is
$F$ and is defined to be a half-open and half-closed time interval. We
denote this model as $\MGen{C}{T}$.

Since there is no penalty for flushing collateral in this model, it is
clear that any algorithm may as well immediately flush any committed
collateral. Despite the simple appearance of the model, it is still
useful to consider this setting as a comparison reference point for $k$-wallet
policies, and we define $\CostF{\A}{\MGen{C}{T}}{\Tx}$ to be the total
value of transactions settled by algorithm $\A$ in this general model
for a transaction sequence $\Tx$ (with total collateral $C$ and maximum
transaction size $T$).

\begin{definition}
  We say that an algorithm $\A$ is $\alpha$-competitive in the general
  collateral model if, for any sequence $\Tx = \tx_1, \ldots, \tx_n$
  with maximum value $T$,
  \[
    \CostF{\OPT}{\MGen{C}{T}}{\Tx} \leq \alpha
    \CostF{\A}{\MGen{C}{T}}{\Tx} + O(1)\,.
  \]
  where the $O(1)$ term may depend on model parameters but not on
  $\Tx$ or $n$.
\end{definition}

Note that for any algorithm $\A$
defined in the $k$-wallet model the following is always true:
\[ \D{C,k}{T}{A} \leq \D{C,k}{T}{\OPT} \leq \G{C}{T}{\OPT}\,. \]

A more natural model arises by introducing a cost for flushes. In
order to reflect the relative cost of flushes in the context of
settled transactions, we introduce two additional parameters:
\begin{enumerate}
	\item \emph{Profit margin $p$}: a profit $p \cdot v$ is gained when a transaction with value $v$ is settled.
	\item \emph{Flush cost $\tau$}: each flush operation costs $\tau$.
\end{enumerate}

We assume throughout that $pC > \tau$; otherwise there is no value to
settling transactions because the cost of even single flush exceeds
the total profit that can be accrued from the flushed collateral.

We let $\MGen{C;\tau}{T;p}$ denote this model, observing that
$\MGen{C}{T}$ and $\MGen{C;0}{T;1}$ coincide. In keeping with the
notation above, we let $\CostF{\A}{\MGen{C;\tau}{T;p}}{\Tx}$ denote
the total profit minus flush cost by applying algorithm $\A$ in the
general model with total collateral $C$, maximum transaction size $T$,
profit margin $p$, flush cost $\tau$, and transaction sequence
$\Tx$. Then, we have the following.

\begin{definition}
  We say that an algorithm $\A$ is $\alpha$-competitive in the general
  collateral model with flush costs if, for any sequence
  $\Tx = \tx_1, \ldots, \tx_n$ with maximum value $T$,
\[
\CostF{\OPT}{\MGen{C;\tau}{T;p}}{\Tx} \leq \alpha \cdot \CostF{\A}{\MGen{C;\tau}{T;p}}{\Tx}  + O(1)\,,
\]
where the $O(1)$ term may depend on the model parameters but not $\Tx$ or $n$.
\end{definition}
  
\noindent\textbf{Transaction size.} Our analysis identifies two regimes of interest regarding transaction costs (for both of the previous models): the ``micro-transaction'' setting, where $T \ll C$ (arising in micropayment applications) and
``arbitrary'' transaction size when $T \approx C$ (arising in more
general settings).

\begin{table}[t!]
  \caption{Summary of our results. Here $r = kT/C$, $\tau$ is the
    flush cost, $p$ is the profit margin, and $f$ is the number of
    flushes.}\label{tab:results}
\begin{tabular}{|lll|}
\hline
\multicolumn{3}{|c|}{\textbf{Discrete k-wallet model}} \\ \hline
\multicolumn{1}{|l|}{\multirow{2}{*}{$r < 1$}} & \multicolumn{2}{l|}{Theorem \ref{theorem: FA in general case}:  $\FA$ is $(2-r)/(1 - r)$-competitive} \\ \cline{2-3} 
\multicolumn{1}{|l|}{} & \multicolumn{2}{l|}{Theorem \ref{theorem: FWF in general case}: $\FWF$ is $(k+1)/(k(1- r))$-competitive} \\ \hline
\multicolumn{1}{|l|}{\multirow{3}{*}{$r = 1$}} & \multicolumn{1}{l|}{$k = 1$} & Theorem \ref{theorem: no deterministic algorithm is competitive when k is 1}: No competitive deterministic settlement algorithm\\ \cline{2-3} 
\multicolumn{1}{|l|}{} & \multicolumn{1}{l|}{\multirow{2}{*}{$k > 1$}} & Theorem \ref{theorem: FA in extreme case}: $\FA$ is 3-competitive \\ \cline{3-3} 
\multicolumn{1}{|l|}{} & \multicolumn{1}{l|}{} &Theorem \ref{theorem: FTWF in extreme case}: $\FTWF$ is $2(k+1)/k$-competitive   \\ \hline
\multicolumn{3}{|c|}{\textbf{General collateral model}} \\ \hline
\multicolumn{1}{|l|}{maximize $V$} & \multicolumn{2}{l|}{Corollary \ref{corollary: FWF is competitive in general model}: $\FWF$ is $(k+1)/(k(1- r))$-competitive} \\ \hline
\multicolumn{1}{|l|}{\begin{tabular}[c]{@{}l@{}}maximize \\ $pV - \tau f$\end{tabular}} & \multicolumn{2}{l|}{\begin{tabular}[c]{@{}l@{}}Theorem \ref{theorem: eta alg competitive ratio}:  $\A_\eta$ is $(1-\beta)/(\sqrt{1-T/C}-\sqrt{\beta})^2$-competitive, \\ where $\eta =  \sqrt{\beta(1-T/C)}$ and $\beta = \tau/pC$\end{tabular}}\\ \hline
\end{tabular}
\end{table}

In the next two sections, we analyze policy competitiveness under each model;  the discrete $k$-wallet model can be found in Section~\ref{sec:fixed} and the general collateral model can be found in Section~\ref{sec:general}. Table~\ref{tab:results} summarizes our results.


\section{The Discrete $k$-Wallet Setting}
\label{sec:fixed}
We now formally consider the $k$-wallet setting. Our focal points are
two natural policies described next: $\FA$ and $\FWF$.\medskip

\subsection{The $\FA$ Algorithm} 
We begin with the simple $\FA$
algorithm, which uses $k$ wallets placed in (arbitrary, but fixed)
order $W_1, \ldots, W_k$. The algorithm packs transactions into its
wallets using the \emph{first fit} algorithm: each transaction is
settled by the first wallet (in the established order) that can fit
the transaction until a transaction arrives that cannot fit into any
wallet. At that time, all $k$ wallets are simultaneously flushed (and so during the flush period $F$ all incoming transactions will be discarded).

In the following theorems, we use $r$ to denote $kT/C$, which is the
ratio between the maximum transaction size and the wallet size. Note
that $r \leq 1$.

\begin{theorem}
\label{theorem: FA in general case}
$\FA$ is $(2-r)/(1 - r)$-competitive in the $\MDisc{C,k}{T}$ model,
where $r = kT/C$.
 \end{theorem}

 \begin{proof}
   For a sequence $\Tx$ of transactions, subdivide time into
   \emph{epochs} according to the behavior of the $\FA$ algorithm. The
   first epoch begins at time $0$ and continues through the first
   flush of the $k$ wallets; the epoch ends in the last timeslot of
   this flush period. Each subsequent epoch begins in the timeslot
   when the wallets come back online (that is, in the timeslot just
   after the previous epoch ends) and continues through the next flush
   to the end of the flush period. In general, there may be a final
   \emph{partial epoch} at the end of the transaction sequence; other
   epochs are referred to as \emph{full}. Any full epoch can be
   further broken into two phases: the \emph{accumulation phase} when
   all transactions are settled by $\FA$, and the \emph{flush phase},
   during which no transactions can be settled (as all wallets are
   offline).

   For any particular full epoch, let $V$ be the total value packed by
   $\FA$ into its wallets in the accumulation phase. We note that
   $V \geq k(C/k-T) = C-kT$, since every wallet will clearly be filled
   to at least $C/k- T$. As for $\OPT$, during the accumulation phase
   it can settle at most $V$ (as this is the value of all transactions
   appearing in that phase) and during the flush phase it can settle
   at most $C$ (as a unit of collateral can settle at most one
   transaction unit in any $F$ period).  Therefore, the ratio between
   the value settled by $\OPT$ and $\FA$ in a full epoch is no more
   than
   \[
     \max_{ C-kT \leq V \leq C} \frac{V+C}{V} \leq \frac{C-kT+C}{C-kT}
     =\frac{(2 - kT/C)}{(1 - kT/C)} = \frac{2-r}{1-r}\,.
   \]

   Moreover, the same formula above can be said
   for any partial epoch, since the accumulation phase comes first.

   Thus, the competitive ratio is $\alpha = (2-r)/(1 - r)$. Observe
   that when $r$ decreases, the competitive ratio approaches 2.
\end{proof}

Aside from the simplicity of the analysis, $\FA$ may have an advantage
for certain sequences of transactions in practice: keeping all $k$
wallets open during the epoch (rather than optimistically flushing
some earlier so as to bring new collateral online earlier) may permit
higher density packing of transactions into the wallets. Indeed, one
could consider leveraging an approximation algorithm for bin packing
for the purposes of optimizing this. On the other hand, in situations
where some of the wallets may become nearly full early in an epoch it
seems wasteful to wait to flush these wallets until all others are
full. This motivates the $\FWF$ algorithm, which attempts to more
eagerly flush wallets so as to bring them online sooner.

\subsection{The $\FWF$ Algorithm} 
We now consider the $\FWF$
algorithm, which fills wallets in a round-robin order. Specifically,
transactions are settled by a particular wallet until a new
transaction arrives that cannot fit; at that point the wallet is
immediately flushed, and the algorithm moves on to the next wallet in
cyclic order. (In cases where the next wallet is offline, the
algorithm waits for the wallet to finish its flush before processing
further transactions, so all transactions arriving during this wait period will be discarded.)

\begin{theorem}
\label{theorem: FWF in general case}
For $k > 1$, $\FWF$ is $(k+1)/(k(1- r))$-competitive in the
$\MDisc{C,k}{T}$ model, where $r = kT/C$.
\end{theorem}

\begin{proof}
Assume, for the purpose of contradiction, that there is a time
   $t$ for which the interval $I = (0, t]$ satisfies
   \[
     V_\OPT (I) > (k+1)/(k(1- r)) \cdot V_\FWFabbr(I)\,,
   \]
   where $V_\OPT (I)$ and $V_\FWFabbr(I)$ are the total values of
   transactions $\OPT$ and $\FWF$ settle during $I$, respectively; let
   $t_e$ be the earliest such $t$.

   Since $t_e$ is the earliest such time, there must be a transaction
   $\tx$ {\em at} $t_e$ that is not settled by $\FWF$.  As $\FWF$ does
   not take $\tx$, it must be the case that either all wallets are
   offline at $t_e$ or $k-1$ wallets are already offline at $t_e$ and
   the remaining wallet goes offline at $t_e$ after failing to fit
   $\tx$.  Therefore, every wallet flushes during
   $I_f = (t_e-F, t_e]$. Suppose, without loss of generality, that
   they do so in order $W_1, W_2, \cdots, W_k$.

   If $t_e \leq F$, then $\OPT$ settles transaction value at most
   $C$ in the interval $(0, t_e]$ since each wallet settles at most
   $C/k$. In the same interval, $\FWF$ settles at least $k(C/k-T)$
   since each wallet settles at least $C/k-T$.  Therefore,
\[
     \frac{V_\OPT (I)} { V_\FWFabbr(I)} \leq \frac{C}{k(C/k-T)} = \frac{C}{C-kT} < \frac{kC}{kC-k^2T} + \frac{C}{kC-k^2T} = \frac{k+1}{k(1-r)}\,,
   \]
which would contradict our assumption.

Otherwise $t_e-F > t_0$. Observe that of the $k$ wallets, at least
$W_2, W_3, \cdots W_k$ began taking transactions during $I_f$ since,
if a wallet $W_i$'s transaction activity before its last flush starts
at a time before $I_f$ for any $i = 2,\cdots, n$, then $W_{i-1}$'s
last flush time must also be before $I_f$ which contradicts the
earlier conclusion that all the $k$ wallets' last flush times are
during $I_f$. Therefore, those $k-1$ wallets together contribute
$(k-1)(C/k- T)$ to $V_\FWFabbr(I_f)$.  The only wallet that may have
started taking transactions before $I_f$ is $W_1$.  Let $t_s$ denote
the last time before $t_e$ that $W_1$ came back online and $t_{s'}$
denote the time $W_1$ flushes.  Note that $t_{s'} \in I_f$, while
$t_s$ may or may not be in the interval.  Let $I_s = (t_s, t_{s'}]$
and $I_{s'} = (t_{s'}, t_e]$; then we have
$V_\FWFabbr(I_s \cup I_{s'}) \geq k(C/k-T)$ since each wallet starts
to take transactions and then flushes within the interval
$I_s \cup I_{s'}$. We also have
$V_\OPT(I_s) \leq V_\FWFabbr(I_s) < C/k$ since wallet $W_1$ is active
during $I_s$.

Additionally, we have $V_\OPT(I_{s'}) \leq C$ since the length of
$I_{s'}$ is no more than $F$, leading to
$V_\OPT(I_s \cup I_{s'}) \leq C/k + C$. Therefore,
  \[
\frac{V_\OPT(I_s \cup I_{s'})}{V_\FWFabbr(I_s \cup I_{s'}) } \leq \frac{C/k + C}{k(C/k - T)} = \frac{k+1}{k} \cdot \frac{C} {C-kT} 
=\frac{k+1}{k(1-r)}\,.
\]

But this contradicts our initial assumption; we conclude that there is
no such $t$.
\end{proof}

\subsection{Optimal Wallet Number} 
When $k$ is large and $r$ is small, $\FWF$ approaches optimality.  For
a given total collateral $C$ and maximum transaction size $T$, it is
natural to ask how many wallets one should choose so as to optimize
the competitive ratio of $\FWF$.  This amounts to determining a $k$
that minimizes $(k+1)/(k(1- kT/C))$.  By computing
\[
\frac{\partial}{\partial k} {\left(\frac{k+1}{k(1-kT/C)}\right)} = 0\,,
\]
we find that the optimal value $k^*$ for $k$ is $\sqrt{1+C/T} - 1$. Of
course, the actual number of wallets must be an integer. We remark
that if $k \approx \sqrt{C/T}$, then each wallet has size
$\approx \sqrt{CT}$ and the competitive ratio is approximately
\[
  \frac{\sqrt{C}+\sqrt{T}}{\sqrt{C}-\sqrt{T}}\,.
\]

\subsection{Remarks on the profit margin--transaction cost setting}
We remark that the competitive analyses above focusing on total settled value immediately give rise to a bound for the setting that introduces a profit margin $p$ and a flush cost $\tau$. Observe that, for any algorithm constrained to the $k$-wallet framework that settles total value $V$, the maximum profit is $V(p - \tau k/C)$, as only $C/k$ value can be settled in any single flush. Thus the profit of $\OPT$ is no more than $V_\OPT(p - \tau k / C)$. On the other hand, the profit of $\FWF$ is at least $V_\FWFabbr(p - \tau k/(C-kT)) - O(1)$, as each wallet is flushed with at least $C/k-T$ value (except for the last wallet, which may introduce a $O(1)$ additional penalty). It follows that the competitive ratio in the profit model is inflated by a factor
\[
  \frac{p/\tau - k / C}{p/\tau - k/(C - kT)}
\]
over that of the ``value-only'' $k$-wallet setting.

Note that the same argument can be applied to $\FA$ because each wallet is likewise flushed with at least $C/k - T$ value (except perhaps for the last flush event).

\subsection{Remarks on the Case $r = 1$}  
If the maximum transaction size can be as large as the wallet size, we make a few additional observations:
\begin{enumerate}
\item No deterministic algorithm can be competitive if there is only one wallet.
\item $\FA$ is 3-competitive.
\item $\FWF$ is not competitive, but a variation on the scheme that groups wallets into pairs can solve the problem.
\end{enumerate}

We prove these in the following.
\begin{theorem}
\label{theorem: no deterministic algorithm is competitive when k is 1}
There is no competitive, deterministic $1$-wallet settlement algorithm
if $r = 1$.
\end{theorem}

\begin{proof}
    For the sake of simplicity, we assume the wallet size and maximum transaction size are both 1. Fixing an online algorithm $\A$, consider 
    the following schedule of transactions:
\begin{itemize}
\item Begin with a rapid succession of one or more {\em microtransactions} each having size $\epsilon$, terminating with the first microtransaction that the algorithm chooses to settle.
\begin{enumerate}
\item If the algorithm does not choose to settle any of the microtransactions, end the succession after $1/\epsilon$ transactions.
\item If the algorithm {\em does} choose to settle one, follow it immediately with a transaction of size 1.
\end{enumerate}
\item Allow an interval of length $F$ to pass without any transactions.
\item Repeat indefinitely.
\end{itemize}

In any iteration of the above, either case 1 or case 2 applies.  In case 1, the online algorithm settles no transactions, while the optimal offline algorithm settles a total value of 1.  In case 2, the online algorithm settles a single transaction worth $\epsilon$ while the optimal offline algorithm settles a single transaction of size 1.  Therefore, the competitive ratio is no better than $1/\epsilon$. As $\epsilon$ can be chosen arbitrarily, it follows that the algorithm cannot achieve any fixed ratio.
\end{proof}

\begin{remark}
  A simple randomized algorithm can achieve constant competitive ratio
  when both $k$ and $r$ are 1.  We first show that $\FA$ with 2
  wallets is 2-competitive against $\OPT$ with one wallet. During each
  epoch, which extends from the time the two wallets come back online
  after the previous flush until the end of the next flush period,
  $\FA$ settles total value $V \geq 1$. On the other hand, $\OPT$ can
  settle at most $V+1$, that is, during the time $\FA$ settles
  transactions, $\OPT$ settles $V$, and during the flush time period of
  $\FA$, $\OPT$ packs 1. Therefore, the competitive ratio is
  $(V+1)/V \leq 2$. Now we will let our randomized algorithm that uses
  one wallet to simulate one of the wallets in the $\FA$ algorithm
  with 2 wallets. At each time when the wallet comes back online, we
  flip a coin, if it is heads, it simulates the first wallet in $\FA$,
  and if it is tails, it simulates the second wallet in $\FA$. That
  is, the wallet in the randomized algorithm only settles the
  transactions that are taken by the chosen wallet and ignores the
  other transactions. The expected value the randomized algorithm can
  pack in each epoch is half of what $\FA$ can pack. Hence the
  competitive ratio against one-wallet $\OPT$ is 4.

\end{remark}

\begin{theorem}
\label{theorem: FA in extreme case}
    For any number $k > 1$ of wallets $\FA$ is 3-competitive if $r = 1$. 
\end{theorem}

\begin{proof}
We use a similar analysis as the proof in Theorem \ref{theorem: FA in general case}.  Time is divided into \emph{epochs}, each of which contains the 
\emph{accumulation phase} and the \emph{flush phase}. For any particular full epoch, let $V$ be the total value packed by $\FA$ into its wallets in the accumulation phase. We note that $V \geq C/2$. To see this, observe that for any pair of wallets $W_i$ and $W_j$ with $i < j$ the final transaction values $v_i$ and $v_j$ of the wallets must satisfy $v_i + v_j > C/k$---otherwise the transactions in the later wallet $j$ would have been placed in the earlier wallet $i$ by first fit. Summing these constraints
\[
    \sum_{i < j} (v_i + v_j) \geq \sum_{i < j} \frac{C}{k} \quad \Rightarrow \quad  
(k-1) \sum_i v_i \geq \frac{k(k-1)}{2}\frac{C}{k} \quad \Rightarrow \quad \sum_i v_i \geq \frac{C}{2}\,.
\]

$\OPT$ can settle at most $V + C$ in this epoch. Considering that $V \geq C/2$, the quantity $V + C \leq 3V$, as desired. It follows that the competitive ratio is $
\alpha \leq 3$ as desired.
\end{proof}

Unfortunately, when $r = 1$, the competitive ratio for $\FWF$ is unbounded.  To see that, again, assume the maximum transaction size and wallet size are both 1. The adversary can produce a series of suitably spaced transactions alternating in value between $\epsilon$ and $1$. $\FWF$ will be forced to take all the $\epsilon$-valued transactions and forgo the high-value transactions, while $\OPT$ can decline to process the low-value transactions in order to process all the high-value ones. Therefore, the competitive ratio would be $1/\epsilon$. This problem can be solved if we pair consecutive wallets and flush each pair when a transaction can not be settled by either of the two wallets.  Within each pair, the second wallet takes a transaction when it is too large for the first wallet. We denote this algorithm as $\FTWF$, for which we have the following result.

\begin{theorem}
\label{theorem: FTWF in extreme case}
When $k > 1$, $\FTWF$ is $2(k+1)/k$-competitive if $r= 1$. 
\end{theorem}

\begin{proof}
The proof is similar to the proof of Theorem \ref{theorem: FWF in general case}.  We use the same notations as before.  Between time interval $(t_0, t]$,
$\FTWF$ can settle transaction value at least $C/2$ since each pair settles at least $C/k$ before they flush, while $\OPT$ settles at most $C+C/k$.  Therefore, 
the competitive ratio is $2(k+1)/k$.
\end{proof}


\section{The General Collateral Setting}
\label{sec:general}
In this section, we study the general model where the entire
collateral $C$ is held in a single pool. A collateral maintenance
policy can replenish any portion of committed collateral (used to
settle a transaction) at any time. Even with this additional
flexibility, a unit of collateral can only be used for settlement once
in a time period of length $F$; it follows that the total settled
value of transactions in any time period of length $F$ is no more than
$C$. Thus, using the same proof as in Theorem~\ref{theorem: FWF in
  general case}, we conclude the following, which shows that $\FWF$ is
competitive even when compared against an adversary who may use the
full power of the general model (while $\FWF$ continues to be
constrained operate in the $k$-wallet discrete model).

\begin{corollary}
  \label{corollary: FWF is competitive in general model}
  Setting $r = kT/C$,
  \[
    \G{C}{T}{\OPT} \leq \frac{k+1}{k(1- r)} \cdot \D{C,k}{T}{\FWF}\,.
  \]
\end{corollary}
 
The above result concerns the total transaction value $V$ settled by
an algorithm. As mentioned in the introduction, without further
constraints on the adversary it's clear that the optimal approach (in
the general model) is to immediately flush any collateral used to
settle a transaction. In practice, this is unattractive as there is,
in fact, a cost associated with the (typically on-chain) transaction
used to refresh collateral. To study this, we introduce two new
parameters: (i.) $p$, the profit margin: the algorithm is provided a
reward of $p \cdot v$ for settling a transaction of value $v$, (ii.)
$\tau$, the cost of any flush (regardless of the amount of collateral
involved in the flush operation).

We seek to maximize the total profit with flush cost deducted.
Formally, we would like to find an algorithm that selects transactions
to settle so that $p \cdot V - \tau f$ is maximized, where $V$ is the
total value of settled transactions and $f$ is the total number of
flushes. (Note that by scaling the figure of merit by $1/\tau$, this
is equivalent to maximizing $(p/\tau)V-f$ and it follows that the
single parameter $p/\tau$ suffices; we separate these merely for the
purpose of intuition.) Recall that we use
$\CostF{\A}{\MGen{C;\tau}{T;p}}{\Tx}$ to denote $p V - \tau f$ for an
algorithm $\A$.

Inspired by the algorithm $\FWF$, we consider a family of policies
that flush when the currently committed collateral has reached a
specified fraction of $C$.

\subsection{The Threshold Algorithm $\A_\eta$}
\newcommand{\Commit}{R} This algorithm is parameterized by a threshold
$\eta$ for which $T/C \leq \eta \leq 1$. The behavior of the algorithm
is determined by the running quantity $\Commit$, the current total
collateral that has been committed to settle transactions, but not
(yet) flushed. The algorithm proceeds as follows: When a new
transaction $\tx$ arrives, it is settled if and only if there is
sufficient remaining collateral. Immediately after settling a
transaction, if $R \geq \eta C$ (so that there is at least $\eta C$
committed but unflushed collateral), then it flushes exactly $\eta C$
collateral.

The following analysis derives the competitive ratio of $\A_\eta$ and
then computes the optimal value of $\eta$, denoted by $\eta^*$, that
minimizes this competitive ratio.

\begin{lemma}
\label{lemma: transaction value for eta alg}
    $\displaystyle\G{C}{T}{\OPT} \leq \frac{C}{C-\eta C - T}\G{C}{T}{\A_\eta}$.
\end{lemma}
 
 \begin{proof}
   The proof is similar to the proof of Theorem~\ref{theorem: FWF in
     general case}, so we are somewhat more brief. For contradiction,
   assume there is a (first) time $t_e$ for which the interval
   $I = (0, t_e]$ satisfies
   \[
     V_\OPT (I) > C/(C-\eta C - T) \cdot V_{\A_\eta}(I)\,,
   \]
   where $V_\OPT (I)$ and $V_{\A_\eta}(I)$ are the total values of
   transactions $\OPT$ and $\A_\eta$ settle during $I$, respectively.

   Since $t_e$ is the earliest such time, there must be a transaction
   $\tx$ {\em at} $t_e$ that is not settled by $\A_\eta$. As $\A_\eta$
   does not take $\tx$, there are two possibilities: 1) all collateral
   is offline at $t_e$; 2) the remaining uncommitted collateral is
   insufficient to settle $\tx$.  Let $I_f = (t_e-F, t_e]$. Recall
   that collateral is flushed sequentially in portions of size
   $\eta C$, and that any such portion will only start to take
   transactions after (or at the same time that) the previous portion
   has been flushed. Let $W_k$, refer to the remaining portion of
   unflushed collateral at time $t_e$, if any, and to the last-flushed
   portion of collateral otherwise.  Let $W_1, W_2, \cdots, W_{k-1}$
   refer to the portions of collateral flushed during all prior flush
   events throughout $I_f$. We have $\sum_{i=1}^k W_i = C$.

   If $t_e \leq F$, then $\OPT$ settles transaction value at most $C$
   in the interval $(0, t_e]$. In the same interval, $\A_\eta$ settles
   at least $C-T$ since the uncommitted collateral is no more than
   $T$.  Therefore,
   \[
     \frac{V_\OPT (I)} { V_{\A_\eta}(I)} \leq \frac{C}{C-T} <
     \frac{C}{C-\eta C - T}\,,
   \]
   which would contradict our assumption.

   Otherwise $t_e-F > t_0$. Observe that of the $k$ portions,
   $W_2, W_3, \cdots,$ and $W_k$ began settling transactions during
   $I_f$ since if a portion $W_i$'s transaction activity before its
   last flush starts at a time before $I_f$ for any $i = 2,\cdots, n$,
   then $W_{i-1}$'s last flush time must also be before $I_f$.
  
   The only portion that may have started settling transactions before
     $I_f$ is $W_1$.  Since $W_1$ has size equal $\eta C$ and the uncommitted collateral
     in $W_k$ is at most $T$, $V_{\A_\eta}(I_f) \geq C - \eta C - T$.

   Again, we have $V_\OPT(I_f) \leq C$ since the length of
   $I_f$ is $F$. Therefore,
   \[
     \frac{V_\OPT(I_f)}{V_{\A_\eta}(I_f)}
     \leq \frac{C}{C-\eta C - T}  \,.
   \]
   This contradicts our initial assumption so we conclude that there is
   no such $t_e$.
 \end{proof}

 \begin{theorem}
   Let $p \in (0,1)$ and $\tau > 0$ be a profit margin and flush
   cost. For a threshold $\eta \in (0,1]$ the algorithm $\A_\eta$ is
   $\alpha$-competitive in the $\MGen{C;\tau}{T;p}$ model for
   \[
     \alpha = \frac{1}{1-\eta - T/C} \cdot \frac{p/\tau-1/C}{p/\tau-1/(\eta C)}\,.
   \]
 \end{theorem}

 \begin{proof}
   For simplicity, assume that at the end of the sequence $\Tx$ any
   committed but unflushed collateral is flushed in both algorithms.
   Note then that the algorithm $\A_\eta$ flushed total collateral
   equal to the total settled value and, furthermore, that each flush
   processes exactly $\eta C$ collateral with the exception of the
   last which may be smaller. It follows that the total number
   of flushes is exactly
   $\lceil{\CostF{\A_\eta}{\MGen{C}{T}}{\Tx}}/{(\eta
         C)}\rceil$.
   We conclude that
   \begin{equation}
     \label{equation: eta alg net profit}
     \begin{split}
       \CostF{\A_\eta}{\MGen{C;\tau}{T;p}}{\Tx} &= p \cdot \CostF{\A_\eta}{\MGen{C}{T}}{\Tx}  -  \tau \cdot \left\lceil \frac{\CostF{\A_\eta}{\MGen{C}{T}}{\Tx}}{\eta C} \right\rceil\\
       &\geq p \cdot \CostF{\A_\eta}{\MGen{C}{T}}{\Tx}  -  \tau \cdot \left( \frac{\CostF{\A_\eta}{\MGen{C}{T}}{\Tx}}{\eta C} + 1\right)\\
                                                &= \CostF{\A_\eta}{\MGen{C}{T}}{\Tx} \left(p - \frac{\tau}{\eta C} \right) - O(1)\,.
     \end{split}
\end{equation}
$\OPT$ flushes at least once when it commits $C$ collateral, therefore
\begin{equation}
  \label{equation: opt net profit}
\CostF{\OPT}{\MGen{C;\tau}{T;p}}{\Tx} \leq p \cdot \CostF{\OPT}{\MGen{C}{T}}{\Tx}  -  \tau \cdot \frac{\CostF{\OPT}{\MGen{C}{T}}{\Tx} }{C} = \CostF{\OPT}{\MGen{C}{T}}{\Tx}(p - \tau/C)\,.
\end{equation}
We combine these to conclude that
\begin{align*}
  \CostF{\OPT}{\MGen{C;\tau}{T;p}}{\Tx} &\leq \CostF{\OPT}{\MGen{C}{T}}{\Tx}(p - \tau/C) \leq \CostF{A_\eta}{\MGen{C}{T}}{\Tx} \frac{C}{C - \eta C - T}(p - \tau/C)\\
  &\leq \CostF{A_\eta}{\MGen{C;\tau}{T;p}}{\Tx} \frac{C}{C - \eta C - T}\cdot\frac{p - \tau/C}{p - \tau/(\eta C)} + O(1)\,,
\end{align*}
as desired. The second inequality holds because of the inequality in Lemma~\ref{lemma: transaction value for eta alg}.
\end{proof}

\noindent\textbf{Optimal value of $\eta$.} The optimal value of $\eta$
(which we denote $\eta^*$) satisfies:
\[
\frac{\partial}{\partial\eta} {\left(\frac{1}{1-\eta - T/C} \cdot \frac{p/\tau-1/C}{p/\tau- 1/(\eta C)}\right)} = 0\,,
\]
which leads to the optimal value $\eta^*$, where $\beta = \tau/(pC)$:
\[
\eta^* = \sqrt{(1 - T/C)\cdot \beta}\,.  
\]

Intuitively, as $\beta$ approaches 0, the flush fee becomes negligible, and the algorithm should flush as often as possible. Using this optimal $\eta^*$, the competitive ratio is $(1-\beta)/(\sqrt{1-T/C}-\sqrt{\beta})^2$, which approaches 1 as $\beta$ approaches 0.  As a final result, we have the following theorem.

 \begin{theorem}
 \label{theorem: eta alg competitive ratio}
 Choosing $\eta = \sqrt{\beta(1-T/C)}$, the competitive ratio for $\A_\eta$ is
 \[
   \frac{1-\beta}{(\sqrt{1-T/C}-\sqrt{\beta})^2}\,.
 \]
  \end{theorem}


\section{Conclusion}
\label{sec:conclusion}
We constructed a modeling framework for collateral management policies of layer-two protocols in the blockchain setting. This framework targets two natural models encountered in practice: the $k$-wallet model in which the collateral $C$ is divided among $k$ wallets, and the general model in which $C$ is viewed as one wallet (or collateral pool). We adopt the standard classical paradigm of competitive analysis in which an online algorithm $\A$, that only knows the transactions encountered so far, is compared against an optimal algorithm $\OPT$ that has full knowledge of the transaction stream including future transactions. Our analysis is agnostic to transaction distribution and only requires knowing the maximum transaction size (i.e., value). Given the dynamic nature of blockchain applications and the unpredictable behavior of their transactions and workload, developing transaction distribution-independent techniques is highly desirable. 

Using our framework, we study natural collateral management policies for the $k$-wallet and the general models, and we show how competitive they are compared to $\OPT$. This is measured in terms of the total transaction value that can be settled and when to replenish the collateral to allow settling future transactions. The general model also studies the replenishment cost and how this affects the utility of the policy. We also derive the optimal configuration for the policy parameters, in terms of the number of wallets and the fraction of the committed collateral to be replenished.

To the best of our knowledge, this work is the first to study the collateral management problem for layer-two protocols. Our future work include extending this model to account for more factors, e.g., transaction specific conditions rather than just a transaction value, and develop dynamic policies in which the number of wallets, and even the collateral value itself, can change over time based on the experienced transaction stream.

\section*{Acknowledgements} 
We thank Mathias Fitzi for conversations that led to the original formulation of these questions. The work of G.A. is supported by NSF Grant No. CNS-2226932.

\bibliographystyle{plain}
\bibliography{walletBib}

\end{document}